\documentclass{llncs}

\usepackage[ruled,vlined,linesnumbered]{algorithm2e}

\usepackage{comment}
\usepackage{amsmath}
\usepackage{graphicx}
\def\PROG#1{$\mathcal{#1}$}

\begin{document}

\title{Consensus on an Unknown Torus with \\
Dense Byzantine Faults}
\author{Joseph Oglio \and Kendric Hood \and Gokarna Sharma \and Mikhail Nesterenko} 
\institute{Department of Computer Science, Kent State University, Kent, OH 44242, USA\\
\email{\{joglio@,khood5@,sharma@cs.,mikhail@cs.\}@kent.edu}}

\maketitle
\thispagestyle{plain}
\pagestyle{plain}

\begin{abstract}
\vspace{-.5cm}
We present a solution to consensus on a torus with Byzantine faults. Any solution to classic consensus that is tolerant to $f$ Byzantine faults requires $2f+1$ node-disjoint paths. Due to limited torus connectivity, this bound necessitates spatial separation between faults. Our solution does not require this many disjoint paths and tolerates dense faults. 

\quad\quad Specifically, we consider the case where all faults are in one column. 
We address the version of consensus where only processes in fault-free columns must agree.  We prove that even this weaker version is not solvable if the column may be completely faulty. 
 We then present a solution for the case where at least one row is fault-free. 
The correct processes share orientation but do not know the identities of other processes or the torus dimensions. The communication is synchronous.

\quad\quad To achieve our solution, we build and prove correct an all-to-all broadcast algorithm \PROG{BAT} that guarantees delivery to all processes in fault-free columns.
We use this algorithm to solve our weak consensus problem. Our solution, \PROG{CBAT}, runs in $O(H+W)$ rounds, where $H$ and $W$ are torus height and width respectively.
We extend our consensus solution to the fixed message size model where it runs in $O(H^3W^2)$ rounds. Our results are immediately applicable if the faults are located in a single row, rather than a column.
\end{abstract}


\section{Introduction}

A Byzantine process~\cite{lamport1982byzantine,pease1980reaching}  may arbitrarily deviate from the prescribed algorithm. This is the strongest fault that can affect a process in a distributed system. The fault is powerful enough to straddle the realm of fault tolerance and security as it may model either a device failure or a malicious intruder. 

In the presence of Byzantine faults, the common task for correct processes is to come to an agreement or consensus. The power of the faults may be abridged with cryptography~\cite{abraham2017brief,dolev1983authenticated,katz2009expected} or randomization~\cite{ben1983another,feldman1997optimal,rabin1983randomized}. If neither primitive is available, the solutions require that the number of correct processes is large enough to overwhelm the faulty ones. 

If the topology of a network is considered, the consensus problem is further complicated as faulty processes, even small in absolute number may isolate some correct processes and prevent them from achieving consensus. In a general topology, it is known that consensus is solvable only if the network is at least $2f+1$-connected~\cite{dolev1982byzantine,fischer1986easy,lamport1982byzantine}. However, such connectivity demands a dense network with large node degree which limits the scalability of a solution achieved this way.

In this paper, we study Byzantine-robust consensus on a torus. Its fixed degree and small diameter makes torus an attractive architecture for distributed computing and storage tasks~\cite{beame1989distributed,ganesan2004one}. 
Torus connectivity is $4$. Hence, according to classic connectivity bounds, it may not tolerate more than $1$ fault. To increase tolerance, the fault locations must be restricted. One approach is to make the faults sparse.  That is, the faulty processes need to be positioned far enough apart so that $2f+1$ connectivity is preserved~\cite{maurer2014byzantine}.  Such a solution fails if the faults are located close to each other. In this paper, we address dense Byzantine faults on a torus of unknown dimensions. 
However, to achieve tolerance, we restrict fault location to a single column.

\vspace{2mm}
\noindent\textbf{Related work.} 
%
There are a number of solutions optimizing consensus in incomplete topologies~\cite{alchieri2008byzantine,chlebus2020fast,nesterenko2006discovering,tseng2015fault,winkler2019consensus}.
Chlebus et al.~\cite{chlebus2020fast} optimize the speed of achieving Byzantine consensus in arbitrary topologies. The connectivity is subject to $2f+1$ bound. 
Alchieri et al.~\cite{alchieri2008byzantine} study synchronous Byzantine consensus with unknown participants. They use participant detectors to establish network membership.
Tseng and Vaidya~\cite{tseng2015fault} explore consensus in directed graphs.
Winkler et al.~\cite{winkler2019consensus} study consensus in directed dynamic networks. 
 Oglio et al.~\cite{oglio2021byzantine} solve Byzantine consensus in Euclidean space.

Potentially, network topology may be discovered despite Byzantine faults. This simplifies consensus solution.
Nesterenko and Tixeuil present two generic Byzantine-tolerant topology discovery
algorithms~\cite{nesterenko2006discovering}. However, neither algorithm is suitable for our task: their first algorithm raises an alarm once an inconsistency is discovered without completing the task. Their second algorithm requires $2f+1$ connectivity.

Let us discuss consensus on toruses and related topologies. 
Several papers~\cite{bhandari2005reliable,koo2004broadcast,pelc2005broadcasting} consider the problem of Byzantine-tolerant reliable broadcast on an infinite grid or a torus in radio networks. In such a network, all processes within a particular distance from the sender receive the message simultaneously. Due to this limitation, their results are not immediately applicable to our model. 

Kandlur and Shin~\cite{kandlur1991reliable} consider a synchronous Byzantine-tolerant broadcast on a torus. In their approach, each message is delivered over a fixed number of node-disjoint paths. The correct message is recovered so long as the majority of paths bypass faulty nodes. Since torus connectivity is $4$, this approach may tolerate at most one fault. Maurer and Tixeuil~\cite{maurer2014byzantine} study a Byzantine-tolerant broadcast on a torus. In their model, the byzantine torus dimensions are not known. Their solution assumes that the faulty nodes are sparse, i.e.  they are located far enough apart such that there is sufficiently many node disjoint paths between the sender and any of the receivers to ensure that the influence of the faulty nodes is countered.

Thus, to the best of our knowledge, previous work has not addressed dense Byzantine faults on a torus. 

\vspace{2mm}
\noindent
\textbf{Paper contribution and approach.} We consider the synchronous model, bidirectional communication, no cryptographic primitives or randomization. The topology we study is a torus whose dimensions: height $H$ and width $W$ are unknown to the processes. All processes share orientation. We assume that all faults are located in a single column. 

We consider the weaker consensus where only processes in fault-free columns must agree on a value. We prove that even this version is not solvable if a single column is completely faulty. 

We examine the case where at least one row is fault-free. To counter faulty column influence, we assume that $W \geq 5$. To solve this weak consensus problem, we first present an all-to-all broadcast algorithm~\PROG{BAT} that guarantees delivery to fault-free columns. We prove it correct and show that it runs in $O(H+W)$ rounds in the LOCAL model~\cite{pease1980reaching}
where messages are of arbitrary size. We then use \PROG{BAT} to build the consensus algorithm \PROG{CBAT}. We prove it correct and show that it runs in $O(H+W)$ rounds in LOCAL.

We then extend \PROG{BAT} and \PROG{CBAT} to use fixed-size messages and estimate their running time in the CONGEST model~\cite{peleg2000distributed}. We determine that the fixed-size message \PROG{BAT} runs in $O(H^2W)$ rounds and \PROG{CBAT} runs in $O(H^3W^2)$ rounds. Our results are immediately applicable to the case of faults located in a single row rather than column. 

Let us now introduce our solution approach. To counter the faults, we leverage the
processes' shared knowledge of the network topology. The data is first propagated along the column and the influence of possibly densely packed column of faults is neutralized. The single correct process is guaranteed to relay data across the faulty column. This horizontally propagated data is then disseminated through the fault-free columns to reach all correct processes there.

\vspace{2mm}
\noindent
\textbf{Paper organization.} In Section~\ref{secNotation}, we introduce our notation, state the problem and prove the necessity of a fault-free row. In Section~\ref{secBAT}, we present and prove correct our all-to-all broadcast algorithm \PROG{BAT}. In Section~\ref{secCBAT}, we use \PROG{BAT} to construct and prove correct our consensus algorithm \PROG{CBAT}. In Section~\ref{secFixed}, we describe how the algorithms can be modified for fixed message size and estimate their run time. We conclude the paper by describing further research directions in Section~\ref{secEnd}.

\section{Notation, Problem Statement, Fault Constraints}\label{secNotation}

\noindent
\textbf{Computation model.}
A process $p$ contains variables and actions. We denote variable $var$ of process $p$ as $var.p$. If process $p$ maintains variable $var$ about process $q$, we denote it as $var.q.p$.
A rectangular grid graph of processes is a Cartesian product of two chain graphs. Such a grid graph is embedded on a plane as a matrix of rows and columns. A \emph{torus grid graph}, or \emph{torus} for short, is formed by starting with a grid graph and connecting corresponding leftmost/rightmost and top/bottom processes with edges. 

Every process in the torus has a unique identifier. An adjacent process is a \emph{neighbor}.
In a torus, every process has exactly four neighbors. Each process knows the identifiers of its neighbors: $\mathit{left}$, $right$, $up$ and $down$. All processes share the orientation. That is, for any two processes $p$ and $q$ if $up.p = q$ then $down.q=p$ and if $\mathit{left}.p = q$ then $right.q = p$. We refer to the shared orientation as North, East, West and South.  The torus dimensions are unknown to the processes. That is, they do not know the height $H$ or the width $W$ of the torus.

The system is completely synchronous. The operation proceeds in rounds. In every round, each process receives all pending messages sent to it, does local calculations and sends messages to its neighbors to be received in the next round.  In one round, a process may thus receive multiple messages from the same or from different neighbors and then send multiple messages to neighbors. A \emph{computation} is an infinite sequence of such rounds.

For most of the paper, we assume that size of the message is arbitrary. We lift this assumption later in the paper.

\vspace{2mm}
\noindent\textbf{Process faults.}
Processes are either correct or faulty. A \emph{correct} process follows the algorithm while a \emph{faulty} process behaves arbitrarily.
A faulty process is always \emph{black}. A correct process is \emph{grey} if it has at most one black process in its column, it is \emph{white} otherwise. The colors of individual processes are applied to the rows and columns of the torus: grey-white, black-white, etc. Refer to Figure~\ref{figBatMessages} for torus depiction and fault location.



\vspace{2mm}
\noindent
\textbf{Broadcast.} Consider the problem where each process $p$ is input an arbitrary initial value $\mathit{initVal}.p$, and 
$p$ must share this value with every other process $q$ so that $val.q.p = \mathit{initVal}.p$ 

\begin{definition}
In the \emph{Weak Synchronous All-to-All Broadcast Problem} every white process $p$ must stop and for each white process $q$, $p$ must output $val.q.p$ such that $val.q.p = \mathit{initVal}.q$. 
\end{definition}

\noindent
\textbf{Consensus.} In \emph{Binary Consensus}, the input value $\mathit{initVal}.p$ is restricted to either $0$ or $1$. Each process must output an irrevocable decision $v$ with the following three properties: 
\begin{itemize}
    \item 
\emph{agreement} -- no two correct processes decide differently; 
\item \emph{validity} -- if there are no faults and for every process $p$, $\mathit{initVal}.p = v$, then $p$ decides $v$; 
\item \emph{termination} -- every correct process eventually decides. 
\end{itemize}
\begin{definition}
In \emph{strong consensus}, the above properties apply to every grey and white process, i.e. to each correct process. In \emph{weak consensus}, the properties apply only to white processes. 
\end{definition}

\noindent
\textbf{Impossibility.} Let us outline the area of the possible. Strong consensus requires $2f+1$ connectivity~\cite{dolev1982byzantine}. The connectivity of torus is $4$, so there is no algorithm that solves strong consensus on a torus with $f > 1$ and arbitrary fault location. 

If faults are restricted to a single column but may occupy the complete column, consensus is still impossible. This holds true even if the processes know the torus dimensions. Intuitively, even if correct processes are connected, they may not be able to distinguish their faulty and non-faulty neighbors and agree on a value. The below theorem formalizes this observation.

\begin{theorem}
There is no algorithm that solves consensus, weak or strong, on a torus with a faulty column.
\end{theorem}
\begin{proof}
Assume the opposite: there is such an algorithm \PROG{A} hat solves weak consensus on a torus with a completely faulty column.  Since there are no grey processes, the requirements of weak and strong consensuses are identical. Consider a torus $T$ of height $1$ and some width $W$. Since the algorithm \PROG{A} is a solution the consensus problem on a torus, it should be able to solve it on $T$. However, this topology is a ring. Its connectivity is $2$. According to Dolev et al.~\cite{dolev1982byzantine}, the consensus is solvable only if the connectivity of the network is at least $2f+1 = 3$. That is, contrary to our initial assumption, \PROG{A} may not solve consensus on $T$. Hence the theorem. 
\qed
\end{proof}

\section{\PROG{BAT}: Byzantine-Tolerant Broadcast to All-to-All on a Torus}\label{secBAT}

\begin{figure}[!tb]
    \centering
    \includegraphics[width=0.50\columnwidth]{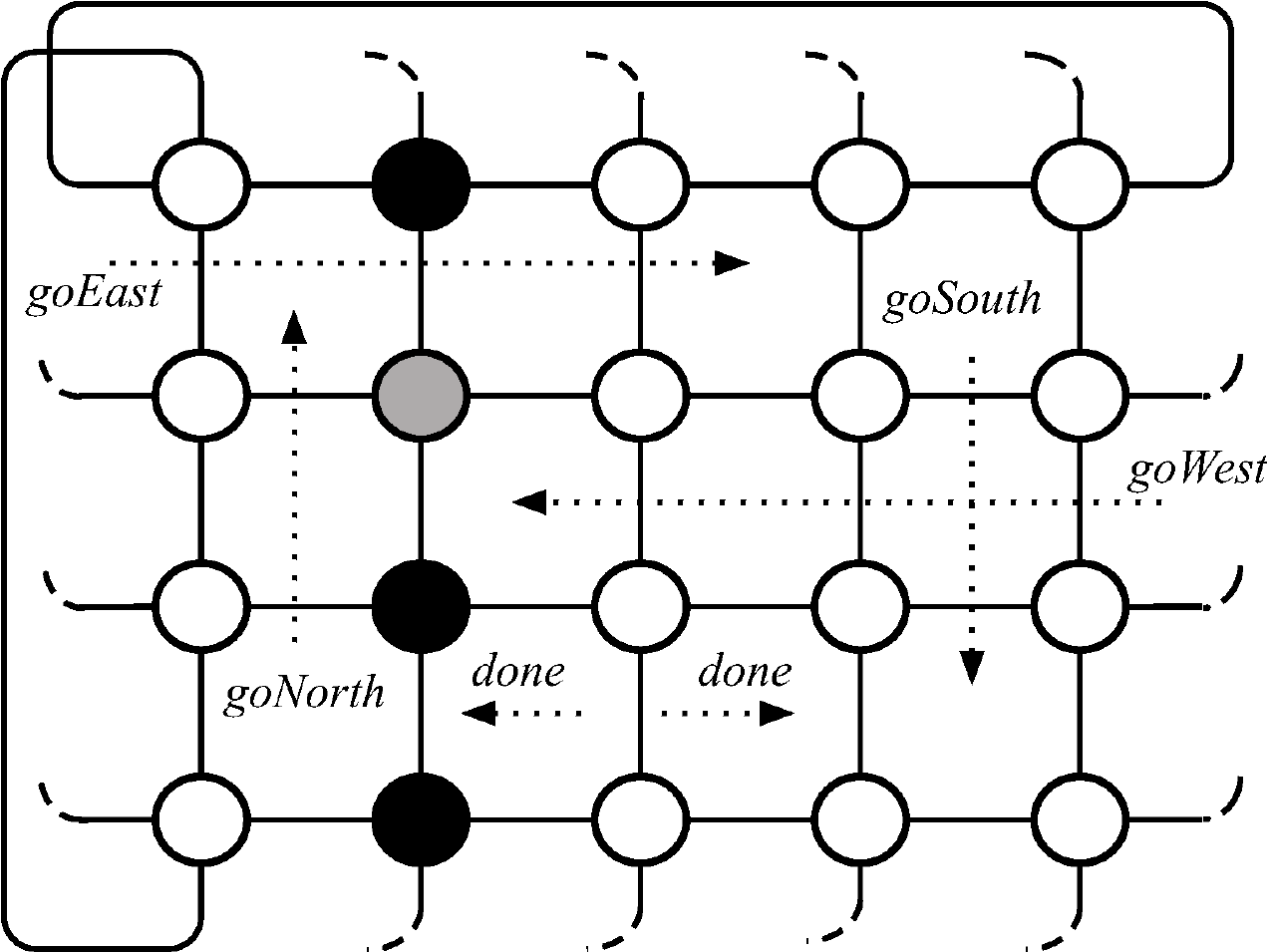}
    \caption{Torus orientation, fault location, message types and message propagation direction in \PROG{BAT}. }
    \label{figBatMessages}
    \vspace{3mm}
\end{figure}

\begin{algorithm}[htbp]
\DontPrintSemicolon
\scriptsize
\SetKwData{receive}{receive}
\SetKwData{send}{send}
\SetKwData{output}{output}
\SetKwData{add}{add}\SetKwData{to}{to}\SetKwData{compute}{compute}
\SetKwData{discard}{discard}
\SetKwData{from}{from}
\SetKwData{deliver}{deliver}
\SetKwData{true}{true}\SetKwData{false}{false}
\SetKwData{bool}{boolean}\SetKwData{int}{integer}

\SetKwFor{Receive}{\receive}{$\longrightarrow$}{}

\caption{\PROG{BAT}: Byzantine All-to-All broadcast on a Torus}\label{algBAT}

\KwIn{$\mathit{initVal}$}

\textbf{Constants:} \\
$p$ \tcp*[f]{process identifier} \\
$up, down, \mathit{left}, right$ \tcp*[f]{neighbor identifiers} \\
\vspace{1mm}
\textbf{Variables:} \\
$northDone$, initially \textbf{false}\\ 
$column$, initially $\langle \rangle$ \tcp*[f]{seq. of value-id pairs received from down neighbor} \\
$\mathit{rowLeft}, rowRight$, initially $\langle \rangle$ \tcp*[f]{columns from resp. left and right neighbors} \\
$matrix$, initially $\langle \rangle$ \tcp*[f]{results matrix to output} \\
$eastWestDone $, initially \textbf{false} \tcp*[f]{one of horizontal neighbors decided}\\
\vspace{1mm}
\textbf{Phases:} \\
\SetKwProg{North}{North}{:}{}
\North(){}{
\textbf{Initial action:}\\
\add $\langle \mathit{initVal},p \rangle$ \to $column$ \tcp*[f]{initiate North Phase} \\
\send $goNorth(\mathit{initVal},p)$ \to $up$  \\
\label{lineInitAlg}
\vspace{1mm}
\textbf{Receive action:}\\
\label{LineSendNorth}

\Receive{$goNorth(v,id)$ \from $down$}{
\If(\tcp*[f]{has not started East-West Phase})
       {$\emph{not}\ northDone$}{
       \eIf{ $id \neq p$ }{ 
            \add $\langle v,id \rangle$ \to $column$ \\
            \send $goNorth(v, id)$ \to $up$
        }
        (\tcp*[f]{$p$ receives its own value back}){
    
            $northDone \leftarrow \textbf{true}$ \\
             \send $goEast(column, \mathit{left}, p, right)$ \to $right$ \\ 
\label{LineStartEastWest}
            \send $goWest(column, \mathit{left}, p, right)$ \to $\mathit{left}$ \\
            }
        }
    }
}

\SetKwProg{East}{East-West, Receive actions}{:}{}
\East(){}{
\Receive{$goEast(c, l, id, r)$ \from $\mathit{left}$}{
       \eIf{$id \neq p$}{ 
         \add $\langle c, l, id, r \rangle$ \to head of $\mathit{rowLeft}$ \\
          \send $goEast(c, l, id, r)$ \to $right$ \\
        }{        
            $m \leftarrow \textbf{match}(\langle column, \mathit{left}, p, right \rangle + \mathit{rowLeft},$ \\
           \hspace{1.6cm} $\langle column, \mathit{left}, p, right \rangle + rowRight$)  \\
\label{LineMatch}
              \If{$matrix = \langle\rangle$ \emph{and} $m \neq \langle\rangle $}{
                $ matrix \leftarrow m$ \\
                 \textbf{Output:} $matrix$ \\
                \send $goSouth(matrix, p)$ \to $down$ \tcp*[f] { start South Phase}
\label{LineStartSouth}\\
 \send $done$ \to $right$ and $\mathit{left}$ \tcp*[f]{start Decision Phase} 

     \lIf
                { $eastWestDone$} {
                \textbf{stop}
                }   
        }
            }
    }
\Receive{$goWest(c, l, id, r)$ \from $right$}{  
    \tcp{handle similar to $goEast$, add $\langle c, l, id, r\rangle$ to tail of $rowRight$ 
    }
 }
}

\SetKwProg{South}{South, Receive action}{:}{}
\South(){}{
\Receive{$goSouth(m,id)$ \from $up$}{
        \If{$id \neq p$}{
           \send $goSouth(m, id)$ \to $down$ \\
           \If(\tcp*[f]{South Phase did not reach $p$} yet)
           { $matrix = \langle\rangle$} {
\label{LineFixMatrix}
                $matrix \leftarrow m$ \\
                 \textbf{Output:} $matrix$ \\
                  \send $done$ \to $right$ and $\mathit{left}$ \tcp*[f]{start Decision Phase}
\label{LineStartDecision}\\
            \lIf
                { $eastWestDone$} {
                \textbf{stop}
                }
            }
        }
    }
}

\SetKwProg{Decision}{Decision, Receive action}{:}{}
\Decision(){}{
\Receive{$done$ \from $direction$}{
        \If{ $\emph{not}\ eastWestDone$ } {
            $eastWestDone \leftarrow \textbf{true}$ \\
\label{LineStop}
            \lIf
                { $matrix \neq \langle\rangle$} {
                \textbf{stop}
            }
        }
    }
}\label{endBat}
\end{algorithm}

\begin{algorithm}[htb]
\setcounter{AlgoLine}{53} 
\scriptsize
\caption{\PROG{BAT} functions}\label{algBATfunctions}
\SetKwData{true}{true}\SetKwData{false}{false}
\textbf{Functions:} \\
\SetKwProg{match}{match($\mathit{rleft}, rright$)}{:}{}

\match(){}{
  $\mathit{cleft} \leftarrow \textbf{consistent}(\mathit{rleft})$ \\
  $cright \leftarrow \textbf{consistent}(rright)$ \\
  \If { $\mathit{cleft} \neq \langle\rangle$ and $cright \neq \langle\rangle $}{
    \tcp{$\mathit{cleft} = cright$}
      \If{\em $|\mathit{cleft}| = |cright| $ and \\
          $(\forall i: 1 \leq i < |\mathit{cleft}|: $ \\
            \ \ $s(i) \equiv \langle lci, lli, lidi, lri \rangle \in \mathit{cleft}$ \\
            \ \  $s(i) \equiv \langle rci, rli, ridi, rri \rangle \in cright$ : \\
            \ \  $lci = rci$ and $lli = rli$ and $lidi = ridi$ and $lri = rri$)}
      {
        \tcp{ return columns of $\mathit{cleft}$}
        \textbf{return} $\langle \forall i:  1 < i \leq |\mathit{cleft}| :
                s(i) \equiv \langle ci, li, idi, ri\rangle \in \mathit{cleft} : ci \rangle$
      }
   }
   \textbf{return} $\langle \rangle$
}

\vspace{2mm}
\SetKwProg{consistent}{consistent($clmn$)}{:}{}

\consistent(){}{

\eIf{\em$s \equiv \langle \cdot, \cdot, id,\cdot \rangle $ is unique in $clmn$ and \\
  \quad exists at most one $i$ : $1 \leq i < |clmn|: cminus = clmn \setminus s(i)$ and  \\
  \quad exists at most one $j$ : $1 \leq j \leq |clmn|: $ \\
\label{LineInsertGrey}
  \quad\quad\quad$cplus = cminus $ insert $\langle \bot, l, id, r \rangle$ at position $j$ in $cminus$ \\
 \quad $\forall i: 1 \leq i <  |cplus|, j = (i+1$ mod $|cplus|):  $ \\
 \quad\quad\quad $s(i) \equiv \langle ci, li, idi, ri \rangle \in cplus$ \\
 \quad\quad\quad $s(j) \equiv \langle cj, lj, idj, rj \rangle \in cplus :$\\
 \quad\quad\quad $idi = lj$ and $ri = idj$
   }{
    \textbf{return} $cplus$
  }{
   \textbf{return} $\langle\rangle$
    }
}
\label{endBATfuncs}
\end{algorithm}

\setlength{\textfloatsep}{20pt}

\noindent
\textbf{Overview.} We present algorithm, we call \PROG{BAT}, that solves the Weak Synchronous All-to-All Broadcast Problem. The code for \PROG{BAT} is shown in Algorithm~\ref{algBAT}. The functions used in \PROG{BAT}  are in Algorithm~\ref{algBATfunctions}. 
Our notation is loosely based on UNITY programming language~\cite{unity}.  
Each process starts \PROG{BAT} by sending initial messages in Line~\ref{lineInitAlg}. The execution of the rest of the actions are \emph{receive actions}. Such an action is  guarded by the corresponding message receipt. It is executed only when this message is sent to the receive process. The actions of the algorithm are grouped into phases: North, East-West, South and Decision. The algorithm is designed such that white processes synchronize their transitions through these phases.

Once done, each process sends the collected column data to its $\mathit{left}$ and $right$ neighbors in the East-West Phase. The data is sent in both directions for verification to counteract the actions of the grey and black processes in the row. If the data received from both directions match, each process starts the South Phase where the confirmed data is sent to $down$ neighbor. This data is a matrix of values from all processes in the torus.

Due to the actions of the black processes in the black-white rows, the white processes may receive corrupted values and fail to complete the South Phase on time. Once the data starts propagating down the column in the South Phase, the white processes synchronize, receive correct matrix and output it. 
After output, processes enter the Decision Phase. This phase ensures termination of white processes.

\vspace{2mm}
\noindent
\textbf{\PROG{BAT} details.} 
The input for each process in the algorithm is an arbitrary value $\mathit{initVal}$.  Each process $p$ in \PROG{BAT} has its own identifier as well as the ids of its up, down, left and right neighbors.

Each process $p$ maintains the following variables. The process has $column$, where it gathers pairs $\langle v, id\rangle$ of values and identifiers of its column in the North Phase. Boolean $doneNorth$ signifies whethe $p$ completed its North Phase. After it is done, the column values are exchanged across the row in two directions: left and right in the East-West Phase. These column values are collected in $\mathit{rowLeft}$ and $rowRight$. Once the values are matched across the row, they are propagated through the column in the South Phase and collected in $matrix$.
Variable $doneNeighbor$ records the neighbor identifiers that made the decision to ensure proper termination.

Let us discuss \PROG{BAT} operation during each phase in detail. Refer to Figure~\ref{figBatMessages} for the messages that processes send and their propagation directions.
All non-black processes simultaneously initiate the North Phase by sending their input values to the $up$ neighbors in a $goNorth$ message (see Line~\ref{LineSendNorth}). Then, each process $p$ starts collecting the values of its column processes by receiving $goNorth$ from the $down$ neighbor and propagating it upward. This continues until $p$ receives its own message back. Once received, $p$ initiates the East-West Phase (Line~\ref{LineStartEastWest}).

For that, $p$ sends the collected column to its $\mathit{left}$ and $right$ neighbor in messages $\mathit{goWest}$ and $goEast$, respectively. Together with the column, $p$ sends its own identifier and the identifiers of $\mathit{left}$ and $right$. This allows the recipients to reconstruct the sequence of the processes in the row if one of the entries is missing. Flag $northDone$ ensures that each process sends $\mathit{goWest}$ and $goEast$ at most once.

Due to the actions of the black processes in its column, even if some grey process completes its own North Phase, it may do so out of sync with the White processes.  However, the receipt of $goEast$ or $\mathit{goWest}$ from its neighbor, enables $g$'s corresponding receive action of the East-West phase. Thus, even though $g$ itself is out of sync, it relays these messages to the white processes in its row allowing them to proceed with the execution the East-West Phase. 

\begin{figure}[!t]
    \centering
    \includegraphics[width=0.56\columnwidth]{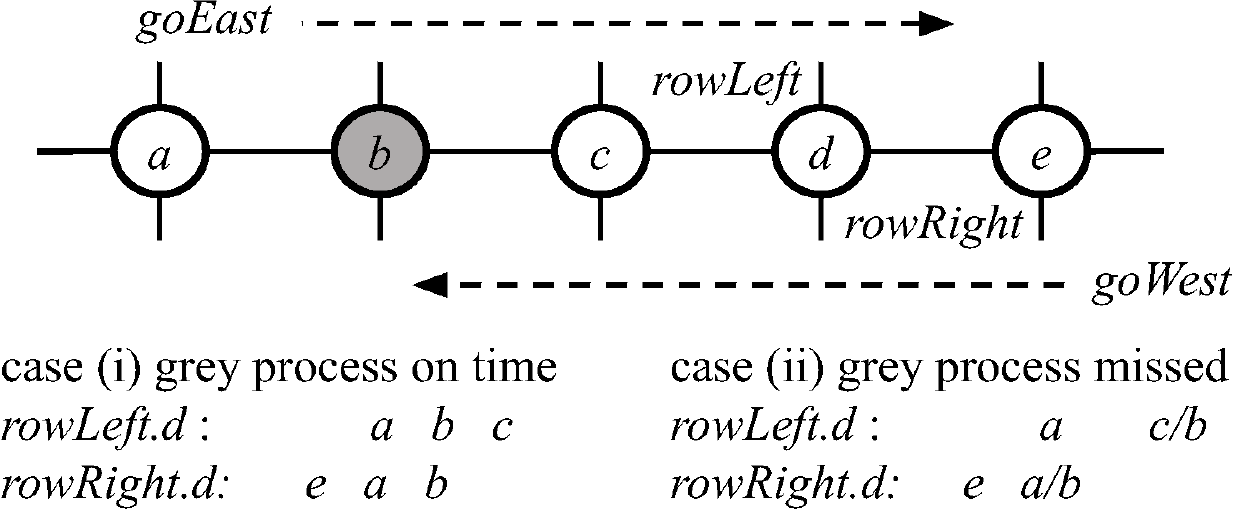}
    \caption{Matrix accumulation during the East-West Phase of \PROG{BAT}. The matrix is accumulated at process $d$ as it collects messages $\mathit{goEast}$ and $\mathit{goWest}$ from its left and right neighbors respectively. In case (i), the grey process $b$ starts East West Phase together with white processes. 
    In case (ii), $b$ starts it one round earlier.}
    \label{figBatRow}
    \vspace{-2mm}
\end{figure}

The East-West Phase is the most complicated part of the algorithm. Let us illustrate its operation with an example. Please see Figure~\ref{figBatRow}. Assume a torus row contains processes $a$, $b$, $c$, $d$ and $e$. All of them are white, except for $b$ which is grey. Each white process completes the North Phase in the same round and starts the East-West Phase by sending the collected column to its right neighbor in a $\mathit{goEast}$ message and to its left neighbor in $\mathit{goWest}$ message. The grey process may (i) complete its North Phase together with the white processes or (ii) miss it and either complete it in a different round or not at all. 

Let us first describe case (i). Once a process receives either $\mathit{goEast}$ or $\mathit{goWest}$, the process records what the message carries and sends it further in the same direction. Thus, in the first round of the East-West Phase, $d$ receives messages from $c$ and $e$, in the second $b$ and $a$ and so on.  As $\mathit{goEast}$ message arrives to $d$ from its left neighbor, $d$ inserts its contents to the head of $\mathit{rowLeft}.d$.

When $\mathit{goWest}$ message arrives to $d$ from the right, $d$ adds its contents to the tail of $\mathit{rowRight}.d$. Figure~\ref{figBatRow} shows the contents $\mathit{rowLeft}.d$ and $\mathit{rowRight}.d$ after three rounds of the East-West Phase when $d$ received messages from $c$, $b,$ and $a$ from the left and $e$, $a$, and $b$ from the right. 
The East-West Phase proceeds until the process receives a message from itself. In the example in Figure~\ref{figBatRow}, this happens in two more rounds, after $d$ receives messages from $e$ then $d$ from the left and $c$ then, possibly, $d$ from the right. 

In case (ii), the grey process $b$ completes the North Phase in a round different from the white processes. In Figure~\ref{figBatRow}, $b$ completes its North Phase one round earlier. Therefore, this message reaches $d$ together with a message from $a$ from the right, and together with message from $c$ from the left.  

After the process receives the message from itself, either from the right or the left,  it compares the contents of $\mathit{rowLeft}$ and $\mathit{rowRight}$ by invoking function  \textbf{match()} (Line~\ref{LineMatch}).
The operation of this function is somewhat complex since grey process may be out of synchrony with the rest of its row. If the contents of the two variables do not match, Function \textbf{match()} returns the data stripped of a mismatched column of the grey process. The resultant list of columns: $matrix$ is the output of \PROG{BAT}. This $matrix$ is sent to the $down$ process in a $goSouth$ message. This begins the South Phase  (Line~\ref{LineStartSouth}).

\PROG{BAT} operation is such that all white processes either initiate the South Phase in the same round or do not initiate the South Phase at all. Indeed, due to the actions of the black processes in the black-white rows, the white processes may not get the matching data and fail to start their own South Phase.  However, the white processes in the grey-white row are guaranteed to start the South Phase. The South Phase started by these processes, re-synchronizes white processes and propagates the correct $matrix$ information in $goSouth$ message.  Once a process receives the missing information, it outputs it (Line~\ref{LineFixMatrix}).

Yet, the processes should not terminate. Indeed, black processes may force a grey process in the black-grey column to receive a $goSouth$ message at any time. The grey processes may output the result and consider its mission accomplished. If a grey process terminates, its halting prevents it from forwarding messages from white processes in the East-West Phase. 

To ensure proper termination, processes execute the Decision Phase. After outputting the decision matrix, each process sends $done$ message to its horizontal, i.e. left and right, immediate neighbors. The process stops if it receives at least one of them (Line~\ref{LineStop}) and it has the decision matrix. One of the horizontal neighbors is guaranteed to be white. Thus, if a process obtains the decision matrix and gets one $done$, then its white neighbors may not be in the middle of the East-West Phase and this process may safely stop. 

\ \\
\textbf{Discussion on \PROG{BAT} functions.} 
Let us now discuss the functions used in \PROG{BAT}. They are shown in Algorithm~\ref{algBATfunctions}. Function \textbf{match()} accepts the input received from the West -- $\mathit{left}$, and East -- $right$ during the East-West Phase. The entries of $\mathit{left}$ and $right$ are of the following format: the column $c$ received by the original sender, its left neighbor identifier $l$, its own identifier $id$ and its right neighbor id $r$. 

First, $\mathit{left}$ and $right$ are individually checked for internal consistency in function \textbf{consistent()}. As process $p$ receives data from either East or West, the left and right neighbors of the subsequent entries should match. In grey-white row, this is true for all entries except for possibly the grey process that may start its East-West phase earlier or later than the white ones. In this case, the grey process entry may arrive out of order. That is, it may arrive together with another white process entry while its own spot in the sequence remains empty. 

Function \textbf{consistent()} finds a potential single entry that is out of sequence as well as a single process gap (Line~\ref{LineInsertGrey}). Specifically, \textbf{consistent}() checks whether there exists at most one element in its parameter $clmn$ such that if it is removed, $cminus$ is produced, and then another element is possibly added to $cminus$ producing $cplus$. This added entry contains $\bot$ for the column values and the $l, id, r$ tuple are such that the left and right neighbors in the consequent entries in resultant $cplus$ match. In this case \textbf{consistent}() returns $cplus$. Consider case (ii) shown in Figure~\ref{figBatRow}. Observe what $\mathit{rowLeft}$ and parameter $\mathit{left}$ contain:
\[
(1)\langle\cdot,c,d,e\rangle, (2)\langle\cdot, d, e, a\rangle, (3)\langle\cdot, e, a, b\rangle, (4)\bot, (5)\langle\cdot, b, c, d\rangle/\langle\cdot, a, b, c\rangle
\]
That is, the third round entry is empty and the fifth round entry contains the data originated by process $c$ and the grey process $b$. Process $b$ sends its data out of sync with the white processes. Function \textbf{consistent}(), on the basis of the adjacent entries, reconstructs the fourth entry and drops the extra fifth entry sent by $b$:
\[
(1)\langle\cdot,c,d,e\rangle, (2)\langle\cdot, d, e, a\rangle, (3)\langle\cdot, e, a, b\rangle, 
(4)\langle\bot, a, b, c\rangle, (5)\langle\cdot, b, c, d\rangle
\]
The corrected entry is stored in $\mathit{cleft}$ by \textbf{match()}. Similar manipulation by \textbf{consistent}() yields $\mathit{cright}$. Function \textbf{match}() determines that $\mathit{cleft}$ is equal to $\mathit{cright}$.  That is, the process $p$ gets the same data from both East and West.
In this case \textbf{match}() returns the matrix of columns stored by $\mathit{cleft}$ and $\mathit{cright}$.

\begin{theorem}\label{thrmBAT}
Algorithm~\PROG{BAT} solves the \emph{Weak Synchronous All-to-All Broadcast Problem} on an unknown torus with Byzantine faults in at most one column with at least one correct row in at most $2H+2+W$ rounds.
\end{theorem}

The above theorem states the correctness and efficiency of \PROG{BAT}. 
See the Appendix for proof details.
Let us discuss the intuition for the correctness proof. It is relatively straightforward to show that all white processes complete the first, North, Phase in $H+1$ round and collect the true contents of their respective columns. By induction on the number of rounds, we show that each process $p$ collects a column of true data about any process $q$ in the same row, so long as the processes between $p$ and $q$, in the direction of information propagation are non-black. Indeed, the non-black processes do not impede data propagation. 

We then separately consider white processes in (a) grey-white rows and (b) black-white rows. In the East-West Phase, the information propagates in two directions. We show that during this phase, in a grey-white row, in round $H+1+W$, the information spreading in two directions matches. The match is up to the potentially misplaced data from the single grey process. We prove that this information is true to the original white process data throughout the matrix. In the black-white row, we show that a white process either completes the phase in $H+1+W$ round with correct white process data or not at all. 

That is,  we show that each white process, regardless of the row location, that complete the East-West Phase, holds correct information about all the white processes in the matrix. Moreover, each white process in the grey-white row is guaranteed to complete this phase. 

We then consider the South Phase. There, we show that it is completed in an additional $H$ rounds, and all the white processes are updated with the correct matrix information before termination. This completes the proof of Theorem~\ref{thrmBAT}.

The computation model that we consider \PROG{BAT} is called LOCAL. It assumes unlimited size messages. Theorem~\ref{thrmBAT} shows that \PROG{BAT} completes in $O(H+W)$ rounds in LOCAL.

\section{\PROG{CBAT}: Consensus Using \PROG{BAT}}\label{secCBAT}

\noindent
\textbf{Description.} Observe that \PROG{BAT} cannot immediately be used to solve consensus since, after its completion, the outputs of each process differ by the values of black and grey processes. An individual white process is not able to determine the color of the senders. Therefore, if the white process makes the consensus decision on the basis of the single execution of \PROG{BAT}, the Byzantine senders may cause white processes to disagree on their outputs. 

Instead we use \PROG{BAT} to select a leader process and agree on the input value of this leader. A particular difficulty is presented if the selected leader process is faulty. In this case, the leader may send different values to different processes. To prevent that, the processes use \PROG{BAT} again to exchange received values, determine whether the leader is consistent in its transmission of its initial value, and, if not, replace the leader . This determination has to proceed despite the inconsistent information provided by the faulty processes in this second exchange. 

\begin{algorithm}[htb]
\setcounter{AlgoLine}{77} 
\scriptsize
\caption{\PROG{CBAT}}\label{algCBAT}
\SetKwData{receive}{receive}
\SetKwData{send}{send}
\SetKwData{output}{output}
\SetKwData{add}{add}\SetKwData{to}{to}\SetKwData{compute}{compute}
\SetKwData{discard}{discard}
\SetKwData{from}{from}
\SetKwData{deliver}{deliver}
\SetKwData{true}{true}\SetKwData{false}{false}
\SetKwData{bool}{boolean}\SetKwData{int}{integer}

\KwIn{$v$ \bool \tcp*[f]{consensus input value} }
\vspace{2mm}

\textbf{Variables:} \\

$M_B \equiv \{ v_{ij}: 1 \leq i \leq H, 1 \leq j \leq W \}$
    \tcp*[f]{matrix of votes received by each process}\\
$M_C \equiv \{ M_{ij}: 1 \leq i \leq H, 1 \leq j \leq W \}$
    \tcp*[f]{votes reported as received by each process}\\
$M_B$ and $M_C$ are initially $\bot$ \\
\vspace{2mm}
\textbf{Stages:} \\
\SetKwProg{Broadcast}{Broadcast}{:}{}
\Broadcast(){}{
    $M_B \leftarrow {\cal BAT}(v)$
}

\SetKwProg{Confirm}{Confirm}{:}{}
\Confirm(){}{
   $M_C \leftarrow {\cal BAT}(M_B)$
}

\SetKwProg{Decide}{Decide}{:}{}
\Decide(){}{
    \tcp{deterministically select leader} 
    $ldr \leftarrow highestID(M_B)$ 
    \label{LineFirstLeader}\\
    let $v_{mn} \in M_B$ be the input value of $ldr$\\
    let $C_L \in M_B$ \tcp*[f]{leader's column in $M_B$}\\
    let $M_L = \{v_{ij}: v_{mn} \in M_{ij} \in M_C \} $  \tcp*[f]{$v$ sent by $ldr$, received by every process } \\
    \ \\
    \tcp{two columns with $\bot$ elements or }
    \tcp{two columns with 1/0 elements or}
    \tcp{two columns whose values differ from rest}
    \If{ $\exists$ $q \neq r$ and $\exists$ $e,f$ : \\
    \quad $v_{qe}, v_{rf} \in M_L\setminus C_L : $\\
    \quad $v_{qe} = v_{rf} = \bot $ \\
    \quad or  \\  
    \quad  $\exists$ $w \neq x$, $y \neq z$ and $\exists$ $a,b,c,d$ : \\
    \quad      $v_{wa}, v_{xb}, v_{yc}, v_{zd} \in M_L \setminus C_L$ : \\
    \quad     $v_{wa} = v_{xb} = 0$ and $v_{yc} = v_{zd} = 1$} {
    \tcp{select a new leader from a different column}    
    $ldr \leftarrow highestID(M_B \setminus C_L)$ 
    \label{LineSecondLeader}\\
    \tcp{recompute $M_L$ and $C_L$ }
    }
    \textbf{Output:} $majority(M_L\setminus C_L)$ 
}

\end{algorithm}

Let us describe the algorithm \PROG{CBAT} that achieves consensus. It is shown in Algorithm~\ref{algCBAT}. We assume that the number of columns $W \geq 5$.
The algorithm has three sequentially executing stages: Broadcast, Confirm and Decide. In the first two stages \PROG{CBAT} executes \PROG{BAT}. 

In the Broadcast Stage, the processes exchange the initial input values. In the Confirm Stage, they transmit the complete matrix they received during Broadcast. At the end of Confirm, each white process receives $H\times W$ matrices of such values.  In the Decide Stage, every white process independently arrives at the consensus decision. 

Let us describe how this decision is computed. Each process $p$ deterministically selects a leader process $ldr$. In the algorithm, it is the process with the highest identifier. Then, out $H \times W$ matrices that process $p$ receives in the second stage, $p$ composes the matrix $M_L$ of values sent to all other processes by the leader as these processes report to $p$. Process $p$ examines these values for inconsistencies. The inconsistencies are as follows: either the values differ in more than one column; or the values contain $\bot$, a sign that \PROG{BAT} detected faulty values, in more than one column; or the values in at least two columns differ from the rest.
In the latter case, since the total number of columns is at least $5$, and the leader column is not considered, there are at least two columns that are different from at least two other columns. 
In all cases, these inconsistencies can be determined by all correct processes. Once the faulty process is detected, the whole column where all the faulty processes may be located is also discovered. Therefore, $p$ selects a new leader from a different column.
Once each process selects the leader, that leader may still be faulty and send arbitrary values to other processes. However, these values are consistently stored in the $M_L$ matrix. Each process decides on the majority of values in this matrix. 

\ \\
\textbf{Correctness proof.}
Note that the below discussion relies on correctness of \PROG{BAT} proved in Theorem~\ref{thrmBAT}.
Denote $val.a$ the input value $\mathit{initVal}$ of some fixed process that is reported by process $a$ after the Broadcast Stage of \PROG{CBAT}. Denote $val.a.b$ the same fixed process value received by process $b$ from process $a$ after the Confirm Stage.

\begin{lemma}\label{lem2stageOK} Let $v$, $w$, $x$ be white processes executing \PROG{CBAT} and distributing the value input to some process $u$. Then, after the Confirm Stage of \PROG{CBAT}, if $u$ is white then $val.v.x = val.w.x$; and 
$val.v.x = val.v.y$ regardless of the color of $u$. 
\end{lemma}
\begin{proof} Let us address the first claim of the lemma where $u$ is white. In the Broadcast Stage, processes are sending their individual input values using \PROG{BAT}. This includes $initVal$ of process $u$.  If the sender $u$ is white and the recipients $v$ and $w$ are white, then, 
according to Theorem~\ref{thrmBAT}, $initVal = val.v = val.w$. During the Confirm Stage, each process, including $v$ and $w$, sends the complete matrix of received values, including $val.v$ and $val.w$ to all processes, including $x$ using \PROG{BAT}. Since $v$, $w$, and $x$ are white, according to Theorem~\ref{thrmBAT}, $val.v = val.v.x$ and $val.w = val.w.x$. That is, $val.v.x = val.w.x$. This proves the first claim of the lemma. 

Let us consider the second claim that should hold regardless of the color of $u$. Consider the value $val.v$ stored at $v$ after the Broadcast Stage. During Confirm, this value is sent to all processes including $x$ and $y$. If $v, x$ and $y$ are white, according to Theorem~\ref{thrmBAT} $val.v = val.v.x = val.v.y$. Hence the second claim of the lemma also holds.\qed\end{proof}

\begin{lemma}\label{lemSameLeader}
Every white process selects the same leader in the Decide Stage of \PROG{CBAT}.
\end{lemma}

\begin{proof}
For each white process $p$,  algorithm \PROG{BAT} is guaranteed to produce a matrix that has the same identifiers, configuration and size. This means that all processes select the same $ldr$ in Line~\ref{LineFirstLeader} of \PROG{CBAT}. 

Let us consider whether each process changes its leader in Line~\ref{LineSecondLeader}. This happens if process $p$ detects inconsistencies in at least two columns excluding the leader column.
Process $ldr$ is either correct or faulty. 

If $ldr$ is correct, then according to the first claim of Lemma~\ref{lem2stageOK}, the entries for $ldr$ in $M_L$ at $p$ are the same, except for possibly a single column of faulty processes. This means that if $ldr$ is correct, none of the white processes changes its leader after initial selection.

If $ldr$ is faulty, then according to the second claim of Lemma~\ref{lem2stageOK}, the corresponding entries for $ldr$ in $M_L$ for all white processes are equal. That is, if processes $p$ and $q$ are white, then if entry $m_{ij} \in M_L$ in  $p$ is equal to $m_{ij} \in M_L$ in $q$. 

That is, if two non-leader columns are inconsistent, then all white processes detect that. That is, if one white process changes leader, all white processes change leaders also.  If the processes change leaders, by the operation of the algorithm, they select the same leader. Hence the lemma.\qed
\end{proof}

\begin{lemma}\label{lemAgreement}
All white processes in \PROG{CBAT} output the same value. If the selected leader $ldr$ is correct, they input $\mathit{initVal}.ldr$.
\end{lemma}
\begin{proof}
Every white process $p$ outputs the majority of values in $M_L \setminus C_L$. Due to Lemma~\ref{lemSameLeader}, white processes in \PROG{CBAT} select the same leader. The leader may be either correct or faulty.

If the leader is correct, then, According to Lemma~\ref{lem2stageOK}, all white entries in $M_L \setminus C_L$ are equal to $\mathit{initVal}.ldr$. Only one column in $M_L \setminus C_L$ may be non-white with potentially arbitrary values. Since we assume that the torus contains at least $5$ columns, the majority of $M_L \setminus C_L$ is equal to $\mathit{initVal}.ldr$ and that is what the white process $p$ outputs.  

Let us consider the case of faulty leader. White processes change the leader only when they detect that the original leader is faulty. They select it from a different column, therefore, the new leader is correct and the previous reasoning to the output value applies. 

The remaining case is that of a faulty leader and no leader change. 
According to the operation of the algorithm, after leader selection at most one column in $M_L\setminus C_L$ has inconsistencies. The rest of the entries hold the same value $x \in {0,1}$. According to the second claim of Lemma~\ref{lem2stageOK}, the inconsistencies are in the same column of $M_L\setminus C_L$ for each white process. This means that the rest of the values in the matrix are $x$. Since white processes output the majority value and the number of columns is at least $5$, all white processes output the same value $x$. \qed
\end{proof}

\begin{theorem}\label{thrmCBAT}
\PROG{CBAT} solves weak consensus on an unknown torus whose width is at least $5$ with Byzantine faults in at most one column and with at least one correct row. 
\end{theorem}

\begin{proof}
We prove the theorem by showing that \PROG{CBAT} satisfies all three properties of consensus.
The agreement property of the consensus requires that all white processes output the same value. Lemma~\ref{lemAgreement} indicates that \PROG{CBAT} satisfies that property. The validity property states that in case there are no faults and all processes are input the same value, they should all output the same value. According to Lemma~\ref{lemAgreement}, if the leader is correct, all correct processes output its input value. Hence, if all processes are correct, and they are all input the same value, they select one process as the leader and output this value. Therefore, \PROG{CBAT} satisfies validity.

Let us address termination. The consensus requires that each correct process terminates. \PROG{CBAT} sequentially executes three stages. The first two stages are executions of \PROG{BAT} and the last one is a finite computation. According to Theorem~\ref{thrmBAT}, \PROG{BAT} terminates. Hence, all three stages of \PROG{CBAT} terminate and this algorithm satisfies the termination property of the consensus.\qed
\end{proof}

Algorithm \PROG{CBAT} sequentially executes \PROG{BAT}. Each \PROG{BAT} completes in $O(H+W)$ rounds in LOCAL. This means that \PROG{CBAT} also completes in $O(H+W)$ rounds.  

\vspace{-3mm}
\section{Extension to Fixed Message Size}\label{secFixed}

As presented, \PROG{BAT} is assumed to operate with unlimited size messages. However, it can be modified to operate with fixed size messages as follows. Observe that the messages get progressively larger as they accumulate the data about the torus. In the first phase, the North Phase, the messages are fixed size since each process only transmits its identifier and its input value. At most one message is sent per link per round. After the completion of this phase, the white processes discover the height $H$ of the torus. In the next phase, the East-West Phase, the processes exchange messages whose size is proportional to the height of the torus. Since white processes know the torus size, this message may be replaced by $H$ fixed size messages transmitted over $H$ rounds. Due to the operation of the algorithm, at most $2$ messages are transmitted per round in the East-West Phase. In the fixed size implementation, each process waits for $2H$ rounds to receive appropriate messages.
The black processes may deceive the grey process and assume the larger torus height. This would make the grey process transmit a message larger than $H$ rounds during the East-West Phase. This, in turn, may prevent the correct message from being transmitted in the same round. To eliminate that, the blocks of the two messages need to be transmitted in the round-robin manner. The South Phase message transmits the complete matrix. So the fixed size implementation has to wait for $H\cdot W$ rounds to receive a single matrix. 

The fixed message size model is called CONGEST.
Let us estimate the running time of the modified algorithm in CONGEST. 
The four phases of the original \PROG{BAT} take $O(H)$, $O(W)$, $O(H)$ and $O(1)$ rounds, respectively. The above argument shows that in each respective phase, the modified \PROG{BAT} needs to send $O(1)$, $O(H)$, $O(HW)$ and $O(1)$ sequential fixed-size messages per round. Hence, the number of rounds in fixed-size message \PROG{BAT} is dominated by the third phase and is in $O(H^2W)$.

Let us now discuss the fixed message size modification of \PROG{CBAT}. In the first stage, \PROG{CBAT} uses \PROG{BAT} to broadcast constant-size decision values. Hence, its run time is in $O(H^2W)$. In the second stage, each process sends a complete $H\cdot W$ matrix. Hence, the run time of this stage, and of the whole algorithm is in $O(H^3W^2)$.

\vspace{-2mm}
\section{Conclusion and Future Work}\label{secEnd}
Algorithm \PROG{BAT} assumes that all the faults are in the same column. In future research, the following conjectures are worth investigating.
We believe that solving the problem with faulty processes spread across multiple columns, one fault per row, is possible but requires substantial modification of the algorithm. We suspect that solving the problem for the case of more than one fault per row is not possible. 

To achieve the solution presented in this paper, we assumed shared process orientation. It is interesting to explore how important this assumption is to the solution. That is, whether it is possible to solve the problem without shared orientation. 
 
In this paper we presented a Byzantine-tolerant consensus algorithm that exploits the knowledge of the network type -- torus, to exceed the tolerance bound presented by the general consensus algorithm.  Our study then opens the following question: what network types have similar properties? To put another way, what specifically makes torus Byzantine-fault resistant and can this property be generalized? We believe that this is a fruitful avenue of future research.

\newpage
\bibliographystyle{abbrv}
\bibliography{torus}

\newpage
\section*{Appendix}

\ \\
\textbf{\PROG{BAT} correctness proof.} We assume that in each computation, the rounds are numbered from $1$ onward.  We say that a process \emph{completes} one of the first three phases (North, East-West, and South) when it starts the next phase by sending the appropriate messages (see Lines~\ref{LineStartEastWest}, \ref{LineStartSouth}, \ref{LineStartDecision}). A process completes the Decision Phase when it executes \textbf{stop} (Line~\ref{LineStop}).

During the computation, each process $p$ records input values of the processes in a column.  We say that this column is \emph{size and id matching to origin} if the column size, order and ids match the actual origin processes. The column is \emph{value matching to origin} if it is size, id and matching to origin and the values of the corresponding entries in the column match the input value of the origin processes. The column is \emph{white value matching to origin}  if it is size, order and id matching and the values of all white processes match the input values of the origin processes. The definitions of a row is stated similarly.

\begin{lemma}\em\label{lemNorthPhaseDone}
{\em In \PROG{BAT}, every white process $p$ completes the North Phase in round $H+1$. After completion, $column.p$ is size, id and value matching to origin.}
\end{lemma}

Observe that, depending on the messages of the faulty neighbors, a grey process may complete the North Phase in the same round with the white process, earlier, later or never. It, however, completes it at most once.

Let $eastRange(p,i)$ be the sequence of $i$ processes in $p$'s row to the East, i.e. right, of $p$. 
Process $p$ itself is not included. The tail of $eastRange$ is the process furthest away from $p$ going East. Let $westRange(p,i)$ be the similarly defined sequence to the West of $p$.

Denote $\mathit{rowLeft}(i)$ the entry in \emph{rowLeft} that was added to it in round $i$. Similarly, $rowRight(i)$ is the entry added to \emph{rowRight} at round $i$. Note that since grey or black process may send horizontal messages out of synch,
other processes may send multiple messages in a single round. Thus,  multiple entries may potentially be added to \emph{rowLeft} in the same round. Also, if a process does not send a message in a particular round, no entry is added then.

\begin{lemma}\label{lemEastWestWhiteOK}
 Let $p$ be a non-black process that completed the North Phase in round $rnd$. Consider
$eastRange(p,i)$ in round $rnd + i$ where $0 \leq i < W$. If $eastRange(p,i)$ is non-black then  (a) for every process $q \in eastRange(p,i)$ it follows that $\langle column.p, \mathit{left}.p, p, right.p \rangle = \mathit{rowLeft}(i).q$  
and (b)  message \\
$\mathit{goEast}(\langle column.p, \mathit{left}.p, p, right.p \rangle)$ is in the outgoing right channel of the tail process of $eastRange(p,i)$.

Similarly, if all processes in $westRange(p,i)$ are non-black, then for every process $q \in westRange(p,i)$, $\langle column.p, \mathit{left}.p, p, right.p \rangle = \mathit{rowRight}(i).q$ and the message $goWest(\langle column.p, \mathit{left}.p, p, right.p \rangle)$ is in the outgoing left channel of the tail process of $westRange(p,i)$.
\end{lemma}

\begin{proof} Let us examine the process actions at the completion of the North Phase. When a process $p$ receives a $goNorth$ message with its own identifier, $p$ sends a message $goEast$ carrying $\langle column.p, \mathit{left}.p, p, right.p\rangle$. 
If $\langle column.p, \mathit{left}.p, p, right.p\rangle$ is present in $\mathit{rowLeft}.p$, then $\mathit{rowLeft}$ is never updated. Each non-black process $q$ forwards $goEast$ sent by $p$, adding the information it carries to $\mathit{rowLeft}.q$. Note that, by the condition of the lemma,  all processes in $eastRange$ are non-black. The rest of the lemma for $eastRange$ is proven by induction on $i$. The proof for $westRange$ is similar.\qed
\end{proof}

A matrix of values, maintained by a process, is \emph{size and id matching to origin} if the identifiers, configuration and size of the processes in the matrix, regardless of process color, matches the torus.  A matrix is also  \emph{white-value matching} if it is size and id matching and the values of all white processes in the matrix match those of the corresponding input values of the origins. 

Let us discuss the relationship between matching to origin of the matrix and its columns. If all the columns of a matrix match to origin, the matrix itself may not necessary do: the order of the columns may be altered or some of the columns may be missing. Hence the following observation. If every white column in the matrix is size, id and value matching to origin and at least one column is size, id and white-value matching to origin, then the matrix itself is size, id and white-value matching to origin. 

\begin{lemma}\label{lemEastWestGreyRowDone}
Every white process $p$ in a grey-white row completes the East-West Phase in round  $H + 1 + W$. After completion, for every white process $p$, $matrix.p$ is size, id and white-value matching to origin.
\end{lemma}

\begin{proof} We prove the lemma by showing that every white column in $matrix.p$ is value matching to origin and the grey-white row itself is white-value matching to origin.

According to the operation of the algorithm, each non-black process originates at most one $goEast$. Other processes always forward it further and it is never forwarded by the originator.  That is, in a grey-white row for each process $p$, the number of entries in $\mathit{eastRange}$ is at most one for each process. In the rest of the proof we discuss the order of these entries.

Let us consider a white process $p$ in a grey-white row in round $rnd = H+1+ W - 1$. 
Let $q$ be another white process that shares a row with $p$.  In the round $rnd$, $p \in eastRange(q,W-1)$. Since $p$'s row is grey-white, $eastRange(q,W-1)$ is grey-white also. Then, according to Lemma~\ref{lemEastWestWhiteOK}, $\langle column.q, \mathit{left}.q, q, right.q \rangle$ is in $\mathit{rowLeft}.p$. According to Lemma~\ref{lemNorthPhaseDone}, the $column.q$ is size, id and value matching to origin. Put another way, the columns in $\mathit{rowLeft}.p$ for white processes in the grey-white row are value matching to origin. Similar discussion applies to $rowRight.p$. 

Let us now discuss the formation of the resultant grey-white row.
Since $p$'s row is grey-white, $eastRange(p,W-1)$ and $westRange(p,W-1)$ are grey-white. According to Lemma~\ref{lemEastWestWhiteOK}, there is a $goEast(\langle column.p, \mathit{left}.p, p, right.p \rangle)$ message in the outgoing right channel of the tail process of $eastRange(p,W-1)$.
However, the width of the torus is $W$. That is, the tail process of  $eastRange(p,W-1)$ is the West (left) neighbor of $p$ itself. This means that this message is in the incoming left channel of $p$. 
Similarly, there is a \emph{goWest} message carrying the same information in the incoming right channel of $p$. 

Let us consider the actions of $p$ when it receives either such $goEast$ or $\mathit{goWest}$. Since the original sender of this message is $p$ itself, $p$ invokes \textbf{match}() with 
$\mathit{rleft} = \langle column.p, \mathit{left}.p, p, right.p \rangle + \mathit{rowLeft}.p$ and $rright = \langle column.p, \mathit{left}.p, p, right.p \rangle + \mathit{rowRight}.p$.
Function  \textbf{match}(), in turn, invokes \textbf{consistent}() separately on $\mathit{rleft}$ and $rright$. 

Function \textbf{consistent}() processes input and then returns $\mathit{cleft}$ and $cright$ respectively. Its operation depends on when the grey process, let us denote it $g$, completes the North Phase.  There are two cases to consider: (1) $g$ completes the North Phase in round $H+1$, (2) $g$ does not complete the North Phase in this round i.e. it either completes it in a different round or not at all. Let us address the two cases individually. 

\begin{enumerate}
\item According to Lemma~\ref{lemNorthPhaseDone}, all white processes in $p$'s row complete their North face in round $H+1$. That is, all processes, including $g$ complete it in this round. 

Let us consider any two neighbor processes $u$ and $v$ that belong to this row. Suppose $v$ is the right neighbor of $u$. That is: $right.u=v$. This means that $\mathit{left}.v = u$. Let $rd(p,u) = i$. In this case $rd(p,v) = i+1$.

According to Lemma~\ref{lemEastWestWhiteOK}, $\langle column.u, \mathit{left}.u, u, right.u \rangle = \mathit{rowLeft}(i).p$ and $\langle column.v, \mathit{left}.v, v, right.v \rangle = \mathit{rowLeft}(i+1).p$. That is, the entries for $u$ and $v$ are consequent in $\mathit{rowLeft}$. This means that for every $i$ :
 $s(i) \equiv \langle ci, li, idi, ri \rangle \in \mathit{rowLeft}$ and $s(j) \equiv \langle cj, lj, idj, rj \rangle \in \mathit{rowLeft}$ it follows that  $idi = lj$ and $ri = idj$. To put another way, the left and right entries of $\mathit{rowLeft}$ match and \textbf{consistent}() returns the rows of these entries for further analysis by \textbf{match}(). Similar discussion applies to $rowRight$.

\item Let us assume that the grey process $g$ completes its North Phase in a round $H+1+j$ other than $H+1$. Note that, depending on whether $g$ completes it earlier or later, $j$ may be either positive or negative. The argument for $g$ not completing its North Phase at all is similar. Let $rd(p,i) = i \neq j$.

According to Lemma~\ref{lemEastWestWhiteOK}, $s(j) \equiv \langle c_j, l_j, id_j, r_j \rangle \in \mathit{rowLeft}(j).p$. 

Each process, including $g$, completes the North Phase at most once.
According to Lemma~\ref{lemEastWestWhiteOK}, all other processes complete this phase in round $H+1$. This means that $\mathit{rowLeft}(i).p$ is empty and it is a single empty spot in $\mathit{rowLeft}$. In this case, \textbf{consistent}() removes $s(j) \equiv \langle c_j, l_j, id_j, r_j \rangle$ and inserts $s(i) \equiv \langle \bot, l_i=id_{i-1}, id_i=l_{i-1}=r_{i+1}, r_i = id_{i+1} \rangle$.
That is, \textbf{consistent}() inserts the entry with an empty column in the same position in both $\mathit{cleft}$ and $cright$.
\end{enumerate}

To summarize, in both cases, \textbf{consistent}() returns entries of $\mathit{goEast}$ and $\mathit{goWest}$, except for possibly, the second case where the out-of-place entry of $g$ is replaced by a properly placed entry. The entries are stored in $\mathit{cleft}$ and $cright$. 

Let us discuss the order of these entries. Again, according to Lemma~\ref{lemEastWestWhiteOK}, the entries are in the order of the processes in this grey-white row. In round $H+1+W-1$, every white process is entered in $\mathit{cleft}$. That is, $\mathit{cleft}$ is size and id matching to origin. 

Let us discuss further operation of \textbf{match}(). It compares corresponding entries $\mathit{cleft}$ and $cright$ for equality. However, according to Lemma~\ref{lemEastWestWhiteOK}, these entries match for white processes and, by operation of \textbf{consistent} for the $g$, they also match for the grey process.

In this case \textbf{match}() returns the columns of $\mathit{cleft}$.  The columns are assigned to $matrix.p$  which ends the East-West phase. This means that this matrix is white-value matching to origin.\qed
\end{proof}

\begin{lemma}\em\label{lemEastWestBlackRowDone}
{\em A white process $p$ in a black-white row completes the East-West Phase only in the round $H + 1 + W$.  If $p$ completes this phase, $matrix.p$ is size, id and white-value matching to origin.}
\end{lemma}

\begin{proof}
Assume $p$ completes the East-West Phase. Let the topology of the row of $p$ be as follows: $\cdots u,x,v\cdots$ where $x$ the black process and $u$ and $v$ are its respective left and right neighbors. 

Let us consider the subsequence $\mathit{cleft}.p$ selected by \textbf{consistent}() for $\mathit{left}/goEast$. This subsequence starts with $p$ and contains all the white processes in the longest white $westRange$ of $p$. This means that it contains $v$. Similarly, the subsequence $cright.p$ selected by \textbf{consistent}() for $right/goWest$ also contains $p$ and the longest white $eastRange$ of $p$, meaning that it contains $u's$. 

According to Lemma~\ref{lemEastWestWhiteOK}, all the white processes entries in $eastRange$ of $\mathit{cleft}.p$ are matching to origin and all the white process entries in $westRange$ of $cright.p$ are also matching to origin.

Entry for $v$ contains $x$ as its right neighbor. Since \textbf{consistent}() selected $\mathit{cleft}.p$, it contains the entry for $x$ where $v$ is its left neighbor. $\mathit{left}.p$ may contain only a single entry for a process identifier. This means that there is only one entry for $x$ in $\mathit{cleft}.p$. Similarly, there is only one entry for $x$ in $\mathit{cright}.p$.

Function $match()$ checks if $\mathit{cleft}.p$ is identical to $cright.p$. 
Since the two are identical, the processes of $eastRange$ of $\mathit{cleft}.p$ have the corresponding entries in $cright$ and match them. The last entry in $eastRange$  of $\mathit{cleft}.p$ is $v$ whose left neighbor is $x$.
Since $\mathit{cleft}.p$ matched $cright.p$ the entry for $x$ must exist in $cright.p$ and match the entry in $\mathit{cleft}.p$. There is only one entry per identifier. Therefore, this entry in $cright.p$ should have the right neighbor of $u$. This and remaining entries in $cright.p$ belong to $westRange$ of $p$. They must match the corresponding entries in $\mathit{cleft}.p$. However, this means that white process entries in $cright.p$ and $\mathit{cleft}.p$ match origin and the entry of $x$ matches the identifier of $origin$. \qed
\end{proof}

\begin{lemma}\label{lemSouthDone}
Every white process $p$ completes the South Phase on or before round $H+1+W+H$. After completion, $matrix.p$ is size, id and white-value matching to origin.
\end{lemma}
\begin{proof}
Each process $p$ maintains the input values in $matrix.p$. The matrix is initially empty and the values are recorded there in one of two ways: (a) by storing the values returned by \textbf{match}() at the completion of the East-West Phase or (b) by receiving the matrix in $gonorth$ message during the South Phase. 

Let us discuss case (a). Due to Lemmas~\ref{lemEastWestGreyRowDone} and~\ref{lemEastWestBlackRowDone}, each white process $p$ completes the East-West Phase in the same round $H + 1 + W$ and with $matrix.p$ that is size, id, and white-value matching origin. That is, the matrices recorded after the completion of the East-West Phase satisfy the lemma. 

Let us consider case (b). After the completion of the East-West Phase, each process $p$ sends the contents of $matrix.p$ down in the $goSouth$ message. In the torus, the columns are either black-grey or white. After round $H + 1 + W$, the only matrices that circulate in white columns in $goSouth$ messages are size, id, and white-value matching to origin. Moreover, according to Lemma~\ref{lemEastWestGreyRowDone}, such message is guaranteed to be sent by every white process in the grey-white row. This means that every process in the white column receives such message in $H$ rounds after it was sent by a process in the grey-white row. That is, every process receives it in $H+1+W+H$ rounds.

Once a process $p$ receives a $goSouth$ message, records the value in the $matrix.p$, forwards the message further down, starts the Decision Phase by sending $done$ messages to left and right neighbors and then, possibly, halts. Thus, the only reason for a process in a white column not to receive the matrix is if it halts. However, by the design of the algorithm, a process may not halt only when the $matrix.p$ is not empty, that is, the process receives the $goSouth$ message. That is, in case (b) each white process eventually stores a matrix of size, id and white-value matching to origin.\qed
\end{proof}

\begin{lemma}\label{lemDecisionDone}
Every white process completes the Decision Phase on or before round $H+1+W+H+1$.
\end{lemma}
\begin{proof}
After the completion of the South Phase, each white process ends a $goSouth$ message South and $done$ messages to East and West. The process completes the Decision Phase when it receives a $goSouth$ phase from its North neighbor and at least one $done$ message from its horizontal neighbors. 

A white process has a white North neighbor and at least one white horizontal neighbor. According to Lemma~\ref{lemSouthDone}, every white process completes the South Phase. This means that each white process receives a $goSouth$ message and at least one $done$ message. This allows each white process to complete the Decision Phase and halt. According to Lemma~\ref{lemSouthDone}, it takes $H+1+W+H$ rounds to complete the South Phase. \qed
\end{proof}

To recap, Lemma~\ref{lemSouthDone} states that the output matrix for each white process $p$ is size, id and white-value matching to origin. That is, for every pair of white processes $p$ and $q$, 
$val.q.p = initVal.q$. Lemma~\ref{lemDecisionDone} that every white process eventually halts. These are the properties of the \emph{Weak Synchronous All-to-All Broadcast Problem}. Moreover, according to Lemma~\ref{lemDecisionDone}, this is accomplished in at most $2H+2+W$ rounds. We summarize the results in Theorem~\ref{thrmBAT}.
\end{document}